\documentclass[10pt,journal, compsoc]{IEEEtran}
\ifCLASSOPTIONcompsoc
  \usepackage[nocompress]{cite}
\else
  \usepackage{cite}
\fi
\ifCLASSINFOpdf
\else
\fi
\usepackage{pgfplots}
\usepackage{amsmath}
\usepackage[utf8]{inputenc}
\usepackage[english]{babel}
\usepackage{graphicx}
\usepackage{tikz}
\usepackage{array}
\usepackage{booktabs}
\usepackage{caption}
\usepackage{tabu}
\usetikzlibrary{shapes,arrows,shadows}
\usepackage{epstopdf}
\usepackage{amsthm}
\newtheorem{theorem}{Theorem}
\newtheorem{lemma}{Lemma}

\usepackage{algorithm}
\usepackage{algpseudocode}
\usepackage{pifont}
\usepackage{subfig}
\usepackage{amssymb}

\theoremstyle{definition}
\newtheorem{definition}{Definition}

\makeatletter
\def\BState{\State\hskip-\ALG@thistlm}
\makeatother

\begin{document}
\title{Absolute Trust: Algorithm for Aggregation of Trust in Peer-to- Peer Networks}

\author{Sateesh~Kumar~Awasthi,
        Yatindra Nath Singh,~\IEEEmembership{Senior~Member,~IEEE,}
        % <-this % stops a space
%\IEEEcompsocitemizethanks{\IEEEcompsocthanksitem M. Shell is with the Department
%of Electrical and Computer Engineering, Georgia Institute of Technology, Atlanta,
%GA, 30332.\protect\\
%E-mail: see http://www.michaelshell.org/contact.html
%\IEEEcompsocthanksitem J. Doe and J. Doe are with Anonymous University.}% <-this % stops a space
%\thanks{Manuscript received April 19, 2005; revised September 17, 2014.}
}

\IEEEtitleabstractindextext{%
\begin{abstract}
To mitigate the attacks by malicious peers and to motivate the peers to share the resources in peer-to-peer networks, several reputation systems  have been proposed in the past. In most of them, the peers evaluate other peers based on their past interactions and then aggregate this information in the whole network. However such an aggregation process requires  approximations in order to converge at some global consensus. It  may not be the true reflection of past behavior of the peers. Moreover such type of aggregation gives only the relative ranking of peers without any absolute evaluation of their past. This is more significant when all the peers responding to a query, are malicious. In such a situation, we can only know that who is better among them without knowing their rank in the whole network. In this paper, we are proposing a new algorithm which accounts for the past behavior of the peers and will estimate the absolute value of the trust of peers. Consequently, we can suitably identify them as a good peers or malicious peers. Our algorithm converges at some global consensus much faster by choosing suitable parameters. Because of its absolute nature it will equally load all the  peers in network. It will also reduce the inauthentic download in the network which was not possible in  existing algorithms.
\end{abstract}

\begin{IEEEkeywords}
Peer-to-peer Network, Trust, Global Trust, Local Trust, DHT, Non-negative matrix, Eigenvector.
\end{IEEEkeywords}}

\maketitle

% is not selected <OR> if conference mode is selected - because compsoc
% conference papers position the abstract like regular (non-compsoc)
\IEEEdisplaynontitleabstractindextext

% \fi
\IEEEpeerreviewmaketitle

\ifCLASSOPTIONcompsoc
\IEEEraisesectionheading{\section{Introduction}\label{sec:introduction}}
\else
\section{Introduction}
\label{sec:introduction}
\fi

\IEEEPARstart{F}{or}  exchanging and sharing the information, peer-to peer networks are better because of their inherent advantage of scalability and robustness, as compare to traditional client server model. Every peer in p2p network can initiate the communication and each peer can act both like client as well as server, and has equal responsibility. But due to lack of functionality of central control, some peers can easily sabotage the network by putting inauthentic contents in the network. Such peers are called malicious peers. Furthermore, rational behavior of peers encourage  them only to draw the resources from network without sharing any thing. These types of peer are called free riders. In such a situation, p2p network functions like a poor client server system where only few peers act as server  with  much less upload bandwidth and storage capacity. The success of peer to peer networks largely depends  on the policy by which these two issues can be handled. Many researchers have proposed to implement  a reputation system based on the past behavior of peers in the network. Past behavior is modelled as trust. It is qualitative and difficult to measure in quantity. Hence, suitable metric and algorithms are required to  measure as well as propagate it to all the nodes, so that behavior of peers can be modelled. \par In most of the existing reputation systems, all the peers evaluate the other peers, based on  the past interactions  and assign them some trust value, also called local trust value. These local trust values are basic information, which are aggregated in whole network to form the global reputation of the peer. This aggregation process is different for structured and unstructured p2p network. In structured network, responsibility to manage global reputation through aggregation is distributed among all the peers. It is also called global trust value of peer. With the help of DHT algorithms, such as Chord\cite{chord}, CAN\cite{can}, Pastry\cite{pastry}, Tapestry\cite{tapestry}, the peer managing global trust of a peer can be easily located. In such a network, consensus is estimated by the manager peer. In an unstructured network, each peer evaluate the global trust value of peers by collecting the local trust from different peers through a distributed aggregation algorithm, the aggregation can be done  either by gossiping protocol or by taking feedback only from  few significant peers; however taking feedback only from  few of the  peers does not make the global trust, global in true sense. In both structured and unstructured network, if consensus is taken across the whole network, then local trust needs to be normalized in some way, which results in approximation of global trust. Some times it may not be the true reflection of the past behavior of peers. This type of aggregation gives only the ranking of peers. This ranking system is similar to the  random surfer Model\cite{random} which is based on the popularity of page on the web. But there is a difference between popularity and trustworthiness; a peer can be popular in a network by doing transaction with large number of peers, but may not  be providing good quality of service.  However, for a trustworthy peer, quality of service provided in each transaction will be good.\par Let us understand it by following example. Let there be five peers in a network - A, B, C, D and E. After some interactions they give some local trust value to each other as shown in  table I. After aggregating these local trust values  \cite{eigen} \cite{power},   they are  ranked as B, E, C, D, A; B is most trustworthy and A is least trustworthy. If it is aggregated as in  \cite{flow},   then they are ranked as E, B, D, C, A; E is most trustworthy and A is least trustworthy. But we can see clearly from Table\ref{table1} that A is making two transactions but both transactions are good as compared to any transaction made by B, C and D. So we can not conclude that A is less trustworthy as compared to B, C and D.\par Furthermore ranking of peers does not give any absolute characterization of them. There are many situation where we need absolute characterization of peers. Consider the example when a peer send a request for a particular file and all the responding peers are malicious, then ranking system can only tells us that who is better among them. We will never know whether they are malicious peers or good peers. It can result in inauthentic download in the network. Besides all these facts, relative ranking of peers overload the most reputable peers even if we use the probabilistic approach to select the source peer for download, i.e.,the probability of choosing a download source peer is proportional  to its global trust. Also, the process of normalization method is  message consuming task. For example, if trust assigning peer updates the local trust of any one of its interacting peer, then it needs to update the values of  all the other interacting peers. It will require more messages to communicate the update to the trust holder peers.
\par Keeping in view  all the above points, a reputation system in p2p network must have the following design considerations. 
\begin{itemize}
\item Reputation should be true reflection of past behavior.
\item Reputation must be aggregated in the whole network.
\item The system should be robust to the malicious peers with as many  attackers model as possible.
\item Load Balance: System should not overload only few peers in network.
\item Adaptive to peer dynamics
\item Fast Convergence Speed
\item Lower overhead/message complexity
\item No Central authority
\end{itemize}

\par In this paper, we  propose a metric and an aggregation algorithm which truly capture the past behavior of the peers. The proposed  aggregation algorithm does not require any kind of normalization hence it automatically meets the above design considerations. It is purely decentralized and does not require any kind of central authority or pre-trusted peers or power nodes. The Absolute Trust is based on the concept of weighted averaging and scaling of local trust. It is calculated  recursively in the whole network till it converges. We will show that it will converge at some unique global value and can be calculated distributively in the whole network by all the peers. Our simulation results show that it gives better authentic downloading performance and more uniform load distributions among good peers with lesser message complexity. 
\begin{table}
\begin{center}
\caption{Local trust of peers A, B, C, D and E, zero means there is no interaction between peers till now }\label{table1}
\begin{tabu} to 0.4\textwidth { | X[c] | X[c] | X[c] | X[c] | X[c]| X[c] | } 
 \hline
    & A& B& C&D & E\\[1ex] 
  \hline
  A &0 & 0.6& 0.6 & 0& 0  \\ [1ex] 
 \hline
  B &0.3 & 0& 0.3 & 0.4& 0.4  \\[1ex] 
   \hline
 C &0.4 & 0.4& 0 & 0.2& 0.2  \\[1ex] 
  \hline
 D&0.5 & 0.1& 0.1 & 0& 0.5  \\[1ex] 
  \hline
 E&0.7 & 0.7& 0.8 & 0& 0  \\[1ex] 
 
  \hline
\end{tabu}
\end{center}
\end{table}

 \par Rest of the paper is organized as follows. Section 2 represents the past work done on reputation systems. In section 3, we will define the basic trust model and its aggregating algorithm. Section 4 will be covering the existence and uniqueness of proposed global trust. In section 5, the algorithm is analyzed. Section 6 presents simulation results, and finally in section 7,  conclusion and future work is presented.

\section{Related Work}

Reputation system is used to establish the trust among the buyers in e-commerce e.g. Amazon, Flipkart, Snapdeal, eBay\cite{eBay}. In all such  systems, there is some central authority and it is keeping the record of past experiences of buyers. This experience is used by new buyers for their shopping. Aggregating the feedback in the presence of central authority is simple task, but p2p system is distributed in nature so maintaining and aggregating the trust is not trivial.\par Aberer and Despotovic \cite{Aberer} proposed a trust model in which only complaints are reported if any, otherwise peers are assumed to be trustworthy.   Eigentrust Algorithm \cite{eigen} is based on Random Surfer Model \cite{random}. Pre-trusted peers are required to handle the malicious peers in it. In PeerTrust \cite{peertrust} five different factors are defined for evaluation of trustworthiness of the peers. Both Eigentrust and PeerTrust are based on the concept of weighted average. Fuzzy Trust model\cite{fuzzy} proposed by Song $et$ $ al$. It is  also based on the concept of weighted average, where  weight factor is determined by three variables-- the peer's reputation, the transaction date and the transaction amount. The message complexity in Fuzzy Trust\cite{fuzzy} is lesser than in the Eigentrust \cite{eigen}. PowerTrust \cite{power} is based on assumption of the power law network. In it local trust is aggregated similarly to the  Eigentrust\cite{eigen} except pre-trusted peers are replaced by most reputable peers in the network. These reputable peers are searched and elected dynamically in the network. All above trust models\cite{Aberer}\cite{eigen} \cite{peertrust} \cite{fuzzy} \cite{power} are for structured network and  DHT is used for efficient location of trust holder peers.  \par In unstructured network, global trust is calculated  by floating the query for local trust in the network. The peer, who wants to calculate the global trust, waits for the feedback upto some time. Then the calculation of global trust  is performed with these limited number of  feedback given by some of the  peers. Gossip Trust\cite{gossip} used same metric as in \cite{eigen} and local trust values are gossiped in the network similarly to randomized gossip algorithm in \cite{random}. In Scalable Feedback Aggregation (SFA)\cite{sfa}, the trustworthiness is calculated by weighted average of local trust and feedback taken by few of the  peers. Antonino $et$ $ al.$ proposed a flow-based reputation\cite{flow} which is modified version of \cite{eigen}. It is only for centralize systems. Wang and Vassileva proposed a Bayesian Trust Model\cite{baysian} in which, different aspect of peer behavior are modelled in different situations. Damiani $et$ $ al$ proposed a system\cite{servent} for managing and sharing the servent's  reputation in which peers poll other peers by broadcasting a request for opinion. In another similar approach\cite{xreputation} Damiani $et$ $ al$. considered the reputation of both peers and resources, but credibility of voter was not considered in both the approaches .
\section{Proposed Trust Model}
In this model of peer to peer network, the peers are assumed to exchange only the files as the resource . With suitable modification, the same model can also be used for other kind of resources. We will define the basic trust metric, namely local trust, which is the raw data used for the calculation of  global trust, which the trust, system as a whole keeps on an individual peer.  Later we will give an algorithm for the aggregation of the local trusts in the whole network to generate global trust value.
\subsection{Local Trust}
Typically peer's satisfaction after a transaction can be classified as satisfied, neutral or unsatisfied. We can also define  many other levels, but for simplicity only three levels have been assumed. Let peer $i$ download some files from peer $j$, then peer $i$ can assign a local trust value to peer $j$ as 

\[T_{ij} = \frac{n_g w_g + n_n w_n + n_b w_b}{n_t}.\] 
\vspace{2mm}
where, \\
      $n_g =$  Number of satisfactroy files,\\
    \hspace{3mm} $n_n =$ Number of average or neutral files,\\
    \hspace{3mm}$n_b =$ Number of unsatisfactroy files,\\           
\hspace{3mm} $n_t = $ Total Number of downloaded files,\\
  \hspace{3mm} $w_g = $ Weigtht factor for  satisfactroy files,\\
  \hspace{3mm}$w_n =$ Weight factor for average or neutral files,\\
 \hspace{3mm} $w_b =$ Weight factor for unsatisfactroy files,\\	
 
 Let us assume that the variation of weight factor varies  linearly from unsatisfactory file to satisfactory file, then
 \[w_n =\frac{w_g + w_b}{2}.\] On simplification,  
 \[T_{ij} =\big[\frac{n_g w_g + (n_t - n_g - n_b)\frac{(w_g+w_b)}{2} + n_b w_b}{n_t}\big]\]
 
 \[T_{ij} = \frac{1}{2 n_t}[(n_g - n_b + n_t )w_g + ( n_b - n_g + n_t ) w_b ]\]

\begin{equation}\label{equ1}
T_{ij} =\frac{1}{2}[(x-y+1)w_g + (y-x+1)w_b ]
\end{equation}

 where,\\
\hspace{3mm}  $x=$  Fraction of satisfactory files, and \\
\hspace{3mm}  $y= $  Fraction of unsatisfactory files.\\
  \par This metric will ensure that local trust value will remain between \(w_g\) and \(w_b\). For example  if peer $i$ download 100 files from peer A and B, and A provide 20 satisfactory files, 40 unsatisfactory files and rest average files  while peer B provide 30 satisfactory files, 60 unsatisfactory files and rest average files; assuming that weight factor of good file is 10 and that is for bad file is 1, then \(T_{iA}=4.6\) and \(T_{iB}=4.15\) \par Many author argue that there are many other factors which can influence the local trust value like amount of transactions, date of transactions, number of transactions etc.\cite{peertrust}\cite{fuzzy}\cite{sfa}\cite{sort}. We  agree with their arguments which can also be considered in our case. But in all the cases, the aggregation process will remain same. In case of free riding, one can define the metric for local trust in many different ways. In the next section, above issues are not discussed and we focus only on the information aggregation in the network. 
  \subsection{Absolute Trust: Algorithm for Aggregation}
 In any evaluation process, there are two parties, one who is evaluating; we will call it the evaluator, and the one who is being evaluated; we will call it the evaluatee. Reliability of evaluation depends on who is evaluating, and it varies from person to person. It is said to be more reliable if it is done by a competent evaluator.\par
 
\vspace{3mm}
\begin{figure}
\centering
\subfloat[One-to-Many]{
 \begin{tikzpicture}[
 roundnode/.style={circle, draw=green!60, fill=green!5, very thick, minimum size=7mm}']
 \draw (5.5,3)circle(.25);
 \draw (4.5,6)circle(.25);
 \draw (5.5,6)circle(.25);
 \draw (6.5,6)circle(.25);
 \draw (5.5,3,0)--(4.5,6);
 \draw (5.5,3,0)--(5.5,6);
 \draw (5.5,3,0)--(6.5,6);
 \end{tikzpicture}
}
\subfloat[Many-to-One]{
\begin{tikzpicture}[
 roundnode/.style={circle, draw=green!60, fill=green!5, very thick, minimum size=7mm}']
 \draw (2,3)circle(.25);
 \draw (3,3)circle(.25);
 \draw (4,3)circle(.25);
 \draw (3,6)circle(.25);
 \draw (2,3,0)--(3,6);
 \draw (3,3,0)--(3,6);
 \draw (4,3,0)--(3,6);
 \end{tikzpicture}
 }
  
\subfloat[One-to-One]{
\begin{tikzpicture}[
 roundnode/.style={circle, draw=green!60, fill=green!5, very thick, minimum size=7mm}']
 \draw (7.5,3)circle(.25);
 \draw (8.5,3)circle(.25);
 \draw (9.5,3)circle(.25);
 \draw (7.5,6)circle(.25);
 \draw (8.5,6)circle(.25);
 \draw (9.5,6)circle(.25);
 \draw (7.5,3,0)--(7.5,6);
 \draw (8.5,3,0)--(8.5,6);
 \draw (9.5,3,0)--(9.5,6);
  \end{tikzpicture}
 }
 \caption{Different ways of evaluation}
 \end{figure}
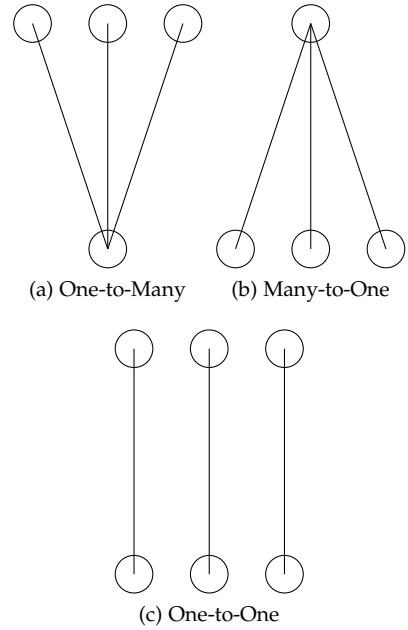 
 
   There are three different scenario in the evaluation as shown in figure 1. One-to-many: one person is evaluating  many persons; many-to-one: many persons are evaluating  one person; and one-to-one: one person is evaluating another person. In one-to-many scenario, since evaluation is done by only one person, the evaluation can be considered to be uniform. Further the evaluated metric can be linearly scaled up or down. In many-to-one evaluation since one person is evaluated by many persons so there are chances of contradictions. At the same time,  the opinion of any evaluator can not be ignored. Thus, the best way to resolve the contradiction is to take the weighted average of all the evaluators'  opinions, while assigning more weight to a more competent evaluator. In one-to-one evaluation, there is no direct comparison of two evaluations because the evaluator and evaluatee both are different. In order to compare these evaluations it is essential to make them uniform with respect to evaluator. Again based on the concept that competent evaluator's evaluation will be more accurate, we can bias these evaluation by a weight factor which must be proportional in some sense  to the  competence of evaluator. This bias can be given by 
 \begin{equation}\label{equ0}
  Eval\_uniform\_out = [ (Eval\_value\_in)^p . ( w_e )^q ]^\frac{1}{p + q}   
 \end{equation}
 where, $Eval\_value\_in$ is evaluation done by an individual evaluator, $w_e$ is weight factor assigned to this evaluator,  $Eval\_uniform\_out$ is output uniform evaluation and $p$  ,  $q$  are suitably chosen constants.
 \par If $p=q$, then $Eval\_uniform\_out$ is geometric mean of $w_e$ and $Eval\_value\_in$. If we take $q=\alpha$.$p$ then \[Eval\_uniform\_out = [ (Eval\_value\_in) . ( w_e )^\alpha ]^\frac{1}{1 + \alpha}\] The impact of $w_e$ and $\alpha$   can be seen in figure 2(a), (b) and (c). In these figures, we can see that $Eval\_uniform\_out$ increases faster with increasing $Eval\_uniform\_in$ for higher values of $w_e$ and for lower values of $\alpha$. The transformation suppresses the reputation reported by less reputed peers, because the suppression is higher for lesser weight factor $w_e$. Also $Eval\_uniform\_out$ is monotonically increasing with $Eval\_uniform\_in$.\par 

%p=2, q=1
\begin{figure}
\centering 
\subfloat[$\alpha$ = 1/2]{
\begin{tikzpicture}
\begin{axis}[
    axis lines = left,
    xlabel = $Eval\_value\_in$,
    ylabel = {$Eval\_uniform\_out$},
    legend pos=north west,
    domain=1:10, 
    samples=10,
   ]
\addplot 
{(x^2 * 1^1)^(1/3)};
\addlegendentry{$w_e=1$}
\addplot
{(x^2 * 2^1)^(1/3)};
\addlegendentry{$w_e=2$}
\addplot
{(x^2 * 3^1)^(1/3)};
\addlegendentry{$w_e=3$}
\addplot
{(x^2 * 4^1)^(1/3)};
\addlegendentry{$w_e=4$}
\addplot
{(x^2 * 5^1)^(1/3)};
\addlegendentry{$w_e=5$}
\addplot
{(x^2 * 6^1)^(1/3)};
\addlegendentry{$w_e=6$}
\addplot
{(x^2 * 7^1)^(1/3)};
\addlegendentry{$w_e=7$}
\addplot
{(x^2 * 8^1)^(1/3)};
\addlegendentry{$w_e=8$}
\addplot
{(x^2 * 9^1)^(1/3)};
\addlegendentry{$w_e=9$}
\addplot
{(x^2 * 10^1)^(1/3)};
\addlegendentry{$w_e=10$}
%\subcaption{hg}
\end{axis}
%\subcaption{hg}
\end{tikzpicture}
}

%\caption{p=2, q=1}
%\end{subfigure}
%\end{figure}
%\hfill
%p=3 q=1
\subfloat[$\alpha$ = 1/3]{
%\begin{figure}
\begin{tikzpicture}
\begin{axis}[
    axis lines = left,
    xlabel = $Eval\_value\_in$,
    ylabel = {$Eval\_uniform\_out$},
    legend pos=north west,
    domain=1:10, 
    samples=10,
     ]

\addplot 
{(x^3 * 1^1)^(1/4)};
\addlegendentry{$w_e=1$}
\addplot
{(x^3 * 2^1)^(1/4)};
\addlegendentry{$w_e=2$}
\addplot
{(x^3 * 3^1)^(1/4)};
\addlegendentry{$w_e=3$}
\addplot
{(x^3 * 4^1)^(1/4)};
\addlegendentry{$w_e=4$}
\addplot
{(x^3 * 5^1)^(1/4)};
\addlegendentry{$w_e=5$}
\addplot
{(x^3 * 6^1)^(1/4)};
\addlegendentry{$w_e=6$}
\addplot
{(x^3 * 7^1)^(1/4)};
\addlegendentry{$w_e=7$}
\addplot
{(x^3 * 8^1)^(1/4)};
\addlegendentry{$w_e=8$}
\addplot
{(x^3 * 9^1)^(1/4)};
\addlegendentry{$w_e=9$}
\addplot
{(x^3 * 10^1)^(1/4)};
\addlegendentry{$w_e=10$}

\end{axis}
\end{tikzpicture}
}
%\caption{p=3,q=1}
%\end{subfigure}
%\end{figure}

%p=4 q=1
\subfloat[$\alpha$ = 1/4]{
%\begin{figure}
\begin{tikzpicture}
\begin{axis}[
    axis lines = left,
    xlabel = $Eval\_value\_in$,
    ylabel = {$Eval\_uniform\_out$},
    legend pos=north west,
    domain=1:10, 
    samples=10,
     ]

\addplot 
{(x^4 * 1^1)^(1/5)};
\addlegendentry{$w_e=1$}
\addplot
{(x^4 * 2^1)^(1/5)};
\addlegendentry{$w_e=2$}
\addplot
{(x^4 * 3^1)^(1/5)};
\addlegendentry{$w_e=3$}
\addplot
{(x^4 * 4^1)^(1/5)};
\addlegendentry{$w_e=4$}
\addplot
{(x^4 * 5^1)^(1/5)};
\addlegendentry{$w_e=5$}
\addplot
{(x^4 * 6^1)^(1/5)};
\addlegendentry{$w_e=6$}
\addplot
{(x^4 * 7^1)^(1/5)};
\addlegendentry{$w_e=7$}
\addplot
{(x^4 * 8^1)^(1/5)};
\addlegendentry{$w_e=8$}
\addplot
{(x^4 * 9^1)^(1/5)};
\addlegendentry{$w_e=9$}
\addplot
{(x^4 * 10^1)^(1/5)};
\addlegendentry{$w_e=10$}

\end{axis}
\end{tikzpicture}
} 
%\end{subfigure}
\caption{Transformation curve, taking evaluated value of trust as a input and uniform evaluated trust value as a output shown for different values of $\alpha$}

\end{figure}
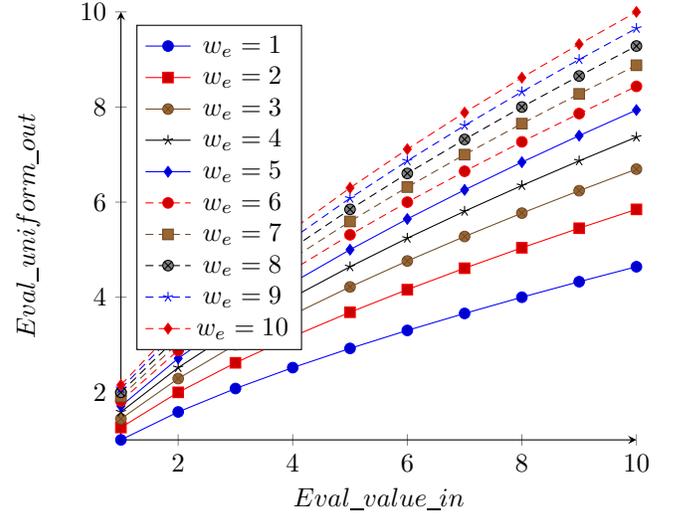
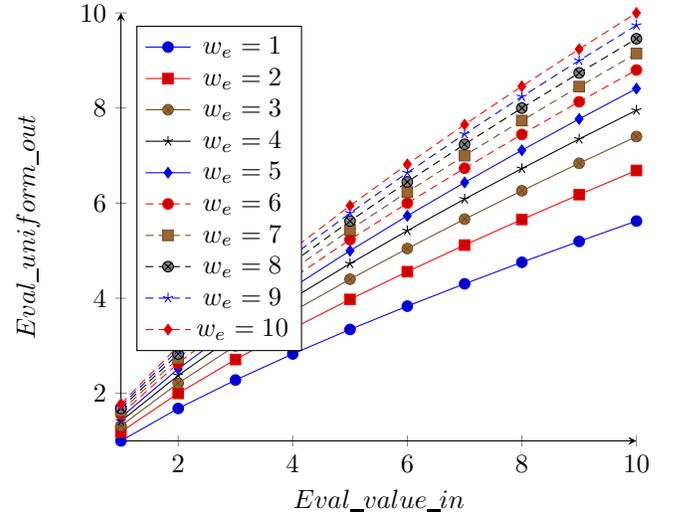
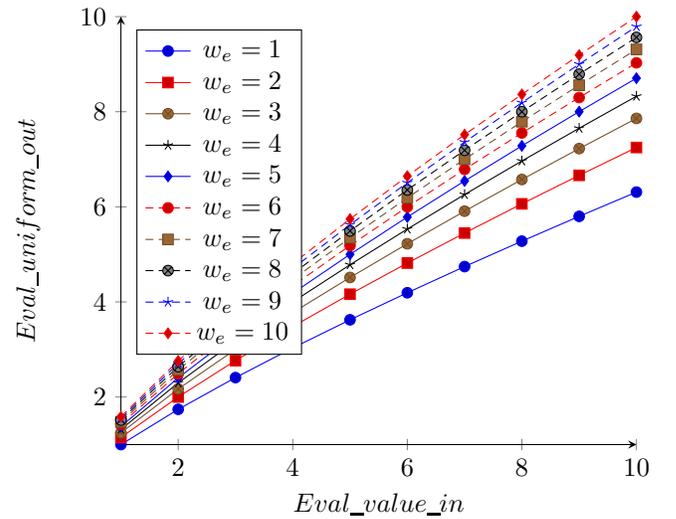
 
   Now let there be $N$ peers in the network, and they are interacting with each other. If any peer $i$ takes the service of any peer $j$ then $i$ can evaluate $j's$ trust according to equation  \ref{equ1}. Each $i$ can evaluate all such $j$ independently and there is no need of any modification in them, because it is one-to-many evaluation. We are aggregating the  values of these one-to-many evaluations (local trust) resulting in estimate of  absolute global trust.\par

 \vspace{3mm}
 \par Each peer is also providing  services to many other peers, and  is being evaluated by them. This is many-to-one evaluation. The aggregated trust values after this step will be weighted average of all the local trust  estimates. The weight factor can be chosen in many different ways, but global trust of an individual peer will be best choice to be used as a weight. Many authors argue that a good service provider may not be a good feedback provider \cite{peertrust}\cite{sfa}\cite{sort}. But we argue that until peers are not in the competition, a good service provider will be most likely a good feedback provider. So we have taken global trust of peers as the weight factor for the purpose of  aggregation of local trust. Hence global trust, $ t_i$ of any peer $i$ is given by

 \begin{equation}\label{equ2}
t_i=\frac{\sum_{j\in S_i}{T_{ji}t_j}}{\sum_{j\in S_i}{t_j}} \hspace{12mm} \forall i.
\end{equation}
  Here, $S_i$ is a set of peers getting services from peer $i$, $T_{ji}$ is  local trust of peer $i$ evaluated by peer $j$, $t_j$ is global trust of peer $j$. Equation \ref{equ2} can be rearranged as 
  \[ t_i=\frac{\sum_{j\in S_i}{T_{ji}t_j}}{\mathbf{ e_i.C.t}}\]
  \[\hspace{13mm} =\sum_{j\in S_i}{(\mathbf{e_iCt})^{-1}T_{ji}t_j}\]
   \par These set of $N$ equations can be written in the form of matrix as 
 % \[t=D^{-1}T^tt\] 
  \[\mathbf{t=[diag(e_1Ct,e_2Ct,.....e_NCt)]^{-1}T^tt}\]
  where, $\mathbf{t}$ is global reputation vector, $ \mathbf{T}$ is trust matrix, its \(T_{ij}\) element is local trust value of peer $j$ assigned by peer $i$. The element \(T_{ij}\) is zero if there is no interaction among peer $i$ and peer $j$, $\mathbf{e_i}$ is row vector with $i^{th}$ entry as 1 and all others are zero, $\mathbf{C}$ is incidence matrix corresponding to $\mathbf{T^t}$ i.e. if $T_{ji}>0$, then $C_{ij}=1$, else $C_{ij}=0$. $\mathbf{T^t}$ is transpose of matrix $\mathbf{T}$.
\par It is clear from equation  \ref{equ2} that value of \(t_i\) will remain between minimum and maximum value of local trust  given by peers belonging to $S$. \par
Now in the whole network every peer  is evaluated by a different set. If the sets can be  represented equivalently as a single peer, then it is same as  one-to-one evaluation. This evaluation can be made uniform using equation \ref{equ0}. To give  the equivalent global trust of the set, consider a set $S$ of $m$ peers with global trust values \(t_1,t_2......t_m\). The global trust of the set must be dominated by the more trustworthy peers because we are giving more weight to their opinion. With the notion of weighted average, intuitively, we can define the global trust of the set as 
  \begin{equation}\label{equ3}
t_s=\frac{\sum_{j\in S}{t_{j}^2}}{\sum_{j\in S}{t_j}}.
\end{equation}
This equation is similar to equation \ref{equ2}. Here, we are ensuring that global trust of a set will be dominated by the peers having higher global trust value. It will always be in between the minimum and maximum values of global trust of the members of set $S$. The global trust, $t_i$, of a peer $i$,  can be biased by the global trust, \(t_{s_i}\), of trust assigning set $S_i$ according to equation \ref{equ0}, then the modified global trust of peer $i$ can be written as 

\[t_i= [{t_i}^p . {t_{s_i}}^q]^\frac{1}{(p+q)}\hspace{1mm} ;\]

\begin{equation}\label{equ4}
t_i=\Bigg[\Bigg(\frac{\sum_{j\in S_i}{T_{ji}t_j}}{\sum_{j\in S_i}{t_j}}\Bigg)^p . \Bigg(\frac{\sum_{j\in S_i}{t_{j}^2}}{\sum_{j\in S_i}{t_j}}\Bigg)^q\Bigg]^\frac{1}{(p+q)}
\end{equation}

Equation \ref{equ4} will be the true reflection of past behavior of peer $i$ in the whole system. This equation will give us the absolute interpretation of global trust value of any  peer. We have made it uniform by using a biasing factor $t_s$. We can now directly compare the global trust values of any two peers. Equation \ref{equ4} can be rearranged as
\[t_i=\Bigg[\Bigg(\frac{\sum_{j\in S_i}{T_{ji}t_j}}{\sum_{j\in S_i}{t_j}}\Bigg). \Bigg(\frac{\sum_{j\in S_i}{t_{j}^2}}{\sum_{j\in S_i}{t_j}}\Bigg)^{q/p}\Bigg]^\frac{1}{(1+q/p)}\]
\hspace{5mm}\[=\Bigg[\Bigg(\frac{(\sum_{j\in S_i}{t_{j}^2})^{q/p}}{(\sum_{j\in S_i}{t_j})^{(1+q/p)}}\Bigg).\Bigg({\sum_{j\in S_i}{T_{ji}t_j}}\Bigg)\Bigg]^\frac{1}{(1+q/p)}\] 
\hspace{3mm}\[=\Bigg[\Bigg(\frac{{(\mathbf{e_i. C. diag(t) .t})}^\alpha} {{(\mathbf{e_i. C. t}})^{(1+\alpha)}}\Bigg) .\Bigg({\sum_{j\in S_i}{T_{ji}t_j}}\Bigg)\Bigg]^\frac{1}{(1+\alpha)}\] 
\hspace{3mm}\[=\Bigg[\sum_{j\in S_i}\Bigg(\frac{{(\mathbf{e_i. C. diag(t) .t})}^\alpha} {{(\mathbf{e_i. C. t}})^{(1+\alpha)}}\Bigg) .\Bigg({{T_{ji}\Bigg) .\Bigg(t_j}}\Bigg)\Bigg]^\frac{1}{(1+\alpha)}\] 

 There are  $N$ nodes in the network, $i=1,2,.....N$, so these set of $N$ equations can be written in the form of matrix as follows \par 
 \[\mathbf{t}=(\mathbf{D . T^t . t})^\frac{1}{1+\alpha}.\]
  $\mathbf{D}$ is a diagonal matrix, with its $i^{th}$ element $d_i$ as $\big[ \frac{{(\mathbf{e_i. C. diag(t) .t})}^\alpha} {{(\mathbf{e_i. C. t}})^{(1+\alpha)}}\big]$, $\mathbf{diag(t)}$ is $NXN$ diagonal matrix, with its $ii^{th}$ element as $t_i$ and  $\alpha = q/p. $ Rest all have same meaning as mentioned above. Power of the vector is defined as the power of its individual element. \par
 
 In this set, there are $N$ unknowns and $N$ non-linear equations, hence we can not state any thing directly about the solution of these equations. However we will show in next section that there exist a unique positive global trust vector, corresponding to these set of equations. We can find the solution iteratively. In each iteration, we are taking a wider view of global trust of any peer $i$ in the network. In subsection 5.1, speed of convergence has been discussed. We can see directly from figure \ref{converge} that the convergence is  more rapid in the initial iterations. This shows that, in the calculation of global trust we are giving more weightage to the one hop neighbors and the weightage decreases as hop counts are increasing.
\section{Existence and uniqueness of Global trust}

We are proposing following lemmas and theorems to show the existence and uniqueness of global trust vector. Few definitions which will be used in this section.\\
\begin{definition}
 A vector $\mathbf{v}$ or matrix $\mathbf{M}$ is said to be positive/nonnegative if its each element $v_i$ or $M_{ij}$ is positive/nonnegative and real.\\
\end{definition}

\begin{definition}
 A vector $\mathbf{v^{'}}$/ matrix $\mathbf{M^{'}}$ is said to be less than  $\mathbf{v^{''}}$/$\mathbf{M^{''}}$ if its each element $v_i^{'}$/$M_{ij}^{'}$ is less than $v_i{''}$/$M_{ij}{''}$. \\
\end{definition}

\begin{lemma}\label{lemma1} 
Let $\mathbf{z}$ be a positive vector in $\mathbb{R}^N$, such that $\mathbf{z}=\mathbf{f(t)}$, with its $i^{th}$ element as $\bigg[\frac {{t_i(\mathbf{e_i C t}})}{(\mathbf{e_i C. diag(t) .t})}\bigg]^{\alpha}.t_i $.  Then 
$\exists$ at least one pair of positive vectors $\mathbf{t^{'}}$ and $\mathbf{t^{''}}$ such that, if $\mathbf{t^{''}}> \mathbf{t^{'}}$, then  \[\mathbf{f(t^{''})}>\mathbf{f(t^{'})}\] where $\alpha$ is an arbitrary rational number.
\end{lemma}

\begin{proof}
Let us consider two vectors $\mathbf{t^{'}}=a.\mathbf{e}$ and $\mathbf{t^{''}}=b.\mathbf{e}$.  Where $\mathbf{e}$ is a vector with all elements as '1', $a$ and $b$ are scalar such that $b>a>0$. Then $i^{th}$ element of vector $\mathbf{f(t^{'})}$ 
   \[f_i(t{'})=\bigg[\frac {{t^{'}_i(\mathbf{e_i C t^{'}}})}{(\mathbf{e_i C. diag(t^{'}) .t^{'}})}\bigg]^{\alpha}.t^{'}_i\]
\[f_i(t{'})=\bigg[\frac {{a(m.a})}{(m.a^2)}\bigg]^{\alpha}.a\]
\[f_i(t{'})=a\] here $m$ is number of $'1'$ in $i^{th}$ row of incidence matrix $\mathbf{C}$. \\Similarly   \[f_i(t{''})=b\] hence $\exists$  a pair of positive vectors $\mathbf{t^{'}}$ and $\mathbf{t^{''}}$ satisfying the condition. 

\end{proof}

\begin{lemma}\label{lemma2}
Let $\mathbf{A}$ and $\mathbf{B}$ be $NXN$ non negative, irreducible matrices with spectral radius '1', and corresponding eigen vector $\mathbf{v}$. Then for any vector $\mathbf{x}$; having at least one component along vector $\mathbf{v}$.
\[\lim_{k\to\infty}(\mathbf{M_1.M_2.M_3........M_k)x} = c.\mathbf{v}\]
Here $\mathbf{M_i}$ can be $\mathbf{A}$ or $\mathbf{B}$ for all $i$ from 1 to $k$ and $c$ is any scalar. $\mathbf{A}$ and $\mathbf{B}$ are such that $(\mathbf{M_1.M_2.M_3........M_k})$ is also irreducible.
\end{lemma}
\begin{proof}
Let the eigen vectors of matrix $\mathbf{A}$ and $\mathbf{B}$ are $\mathbf{v,v_2,v_3,......v_N}$ and $\mathbf{v,u_2,u_3,.....v_N}$. Then any vector $\mathbf{x}$; having at least one component along vector $\mathbf{v}$, can be expressed as \[\mathbf{x}=a_1\mathbf{v}+a_2\mathbf{v_2}+.......a_N\mathbf{v_N}\] and \[\mathbf{x}=b_1\mathbf{v}+b_2\mathbf{u_2}+......b_N\mathbf{u_N}\]when this vector will pass through matrix $\mathbf{A}$ and $\mathbf{B}$ then it will be \[\mathbf{Ax}=a_1\mathbf{v} + a_2\lambda_2 \mathbf{v_2}+.....+a_N\lambda_N\mathbf{v_N}\] and \[\mathbf{Bx}=b_1\mathbf{v}+b_2\gamma_2\mathbf{u_2}+......+n_N\gamma_N\mathbf{u_N}\] where
$\lambda_2,\lambda_3......\lambda_N$ and $\gamma_2,\gamma_3......\gamma_N$ are eigen values of matrix $\mathbf{A}$ and $\mathbf{B}$ respectively. If it will again pass through any of $\mathbf{A}$ and $\mathbf{B}$, then vector $\mathbf{v}$ will remain as it is and magnitude of all other vectors will decrease because $1>|\lambda_2|>|\lambda_3|....>|\lambda_N|$ and $1> |\gamma_2|>|\gamma_3|......>|\gamma_N|$ (see \cite{nonnegative}). Thus \[\mathbf{BAx} = a_1\mathbf{v} + \delta a \mathbf{v} + L.O.M.O. \mathbf{u_2,u_3,....u_N}\] and \[\mathbf{ABx} = b_2\mathbf{v} + \delta b \mathbf{v} + L.O.M.O \mathbf{v_2,v_3,......v_N}\] $L.O.M.O.$ means "lower order magnitude of". Repeating this operation $k^{th}$ times in any order we will get \[\lim_{k\to\infty}(\mathbf{M_1.M_2.M_3........M_k)x} = c.\mathbf{v}\] 
where $\mathbf{M_i}$ can be $\mathbf{A}$ or $\mathbf{B}$ for all $i$ from 1 to $k$

\end{proof}

\begin{theorem}\label{1}
Let $\mathbf{A}$ and $\mathbf{B}$ be $NXN$ non negative, irreducible matrices with spectral radius '1', and corresponding eigen vector $\mathbf{v}$. Then for any vector $\mathbf{x}$ in $\mathbb{R}^N$
 
\[\lim_{k\to\infty} (\mathbf{A-B})^k \mathbf{x} = \mathbf{0}\]  
\end{theorem}

\begin{proof}
In  Lemma \ref{lemma2} let $\mathbf{M_i}= \mathbf{A}$ for all $i$,  then
\begin{equation}\label{equ5}
\lim_{k\to\infty} \mathbf{A}^{k-1} \mathbf{x} \approx a_1 \mathbf{v},
\end{equation} 
and if $\mathbf{M_i}= \mathbf{B} $ for all $i$, then
\begin{equation}\label{equ6}
\lim_{k\to\infty} \mathbf{B}^{k-1} \mathbf{x} \approx a_2 \mathbf{v} 
\end{equation}
 if $\mathbf{M_i}$ is taken randomly $\mathbf{A}$ or $\mathbf{B}$, then
 \begin{equation}\label{equ7}
\lim_{k\to\infty}(\mathbf{A.B......B.A...}(k-1) times    )\mathbf{x} \approx a_3 \mathbf{v}
 \end{equation}
 where $a_1, a_2$ and $ a_3 $ are some scalers, adding \ref{equ5}, \ref{equ6} with all combinations of \ref{equ7} will result
 \begin{equation}\label{equ8}
 \lim_{k\to\infty}(\mathbf{A-B})^{k-1} \mathbf{x} \approx b \mathbf{v}
 \end{equation}
here $b$ is a linear combination of $a_1, a_2$ and all $ a_3 $. Now  pre-multiplying equation 
\ref{equ8} by $(\mathbf{A-B})$,  \[\mathbf{(A-B)(A-B)}^{k-1}\mathbf{x} = \mathbf{(A-B)} b\mathbf{v}=(\mathbf{v-v})b=\mathbf{0}\]  
Hence 
\[\lim_{k\to\infty} (\mathbf{A-B})^k \mathbf{x} = \mathbf{0}\]  
\end{proof}

\begin{theorem}\label{0}
Let $\mathbf{A,A_1,A_2......A_m}$ be $NXN$ non negative, irreducible matrices, with spectral radius $1,\lambda_{1},\lambda_{2}......\lambda_{m}$ respectively. Let the corresponding eigen vector for all the above matrices be $\mathbf{v}$. Then for any vector $\mathbf{x}$.
\[\lim_{k\to\infty}(\mathbf{A_1+A_2+.....A_m-A})^k \mathbf{x}=\mathbf{0}\]
if $|\lambda_{1}+\lambda_{2}+......\lambda_{m}-1| < 1$
\end{theorem}

\begin{proof}
Let \[\mathbf{M}=(\mathbf{A_1+A_2+.......+A_m})\]
then 
\[\mathbf{M.v} =(\mathbf{A_1+A_2+.......+A_m).v}\]
    \[=(\lambda_{1}+\lambda_{2}+......+\lambda_{m}).\mathbf{v}=\lambda.\mathbf{v}\]
hence $\mathbf{v}$ is also an eigen vector of matrix $\mathbf{M}$ and corresponding eigen value is $ \lambda$. Matrix $\mathbf{M}$ is the sum of non negative, irreducible matrices $\mathbf{A_1,A_2......A_m}$ therefore $\mathbf{M}$ is also non negative and irreducible. So we can conclude that spectral radius of matrix $\mathbf{M}$ is  $ \lambda$.\par Further $\mathbf{M}$ can be written as
 \[\mathbf{M}=\bigg[\frac{\mathbf{M}}{\lambda}+\frac{(\lambda-1)\mathbf{M}}{\lambda}\bigg]\]
    \[=[\mathbf{B} + \mathbf{N}]\]   
here $\mathbf{B}$ is $\mathbf{M}/\lambda$ and $\mathbf{N}$ is $(\lambda-1)\mathbf{M}/\lambda$. Matrix $\mathbf{B}$ and $\mathbf{N}$ are scalar multiple of non negative irreducible matrix $\mathbf{M}$ therefore matrix $\mathbf{B}$ and $\mathbf{N}$  also follow the properties of  non negative irreducible matrices. Hence spectral radius of  matrix $\mathbf{B}$ and $\mathbf{N}$ is '1' and $|\lambda-1|$ respectively and corresponding eigen vector is $\mathbf{v}$.\par If $|\lambda-1|<1$ then 
\begin{equation}\label{9}
\lim_{k\to\infty}\mathbf{N}^k\mathbf{x}=\mathbf{0},
\end{equation}
  and from Theorem \ref{1}  
\begin{equation}\label{10}
\lim_{k\to\infty} (\mathbf{B-A})^k \mathbf{x} = \mathbf{0}.
\end{equation}  
  
In fact when vector $\mathbf{x}$ is passed through any of $\mathbf{N}$ or $(\mathbf{B-A})$ its magnitude decreases, and at $k\to\infty$, it become zero. So in general we can write 
\begin{equation}\label{11}
\lim_{k\to\infty}\mathbf{(M_1.M_2.M_3........M_k)x} = \mathbf{0}
\end{equation}
 where $\mathbf{M_i}$ can be any of $\mathbf{N}$ or $(\mathbf{B-A})$. Adding all the combinations of equation \ref{11} with equation \ref{9} and equation \ref{10}, we will get 
  \[\lim_{k\to\infty}\mathbf{(N+(B-A))}^k\mathbf{x}=\mathbf{0}\] or
   \[\lim_{k\to\infty}(\mathbf{M-A})^k\mathbf{x}=\mathbf{0}\] hence
   \[\lim_{k\to\infty}(\mathbf{A_1+A_2+.....A_m-A})^k \mathbf{x}=\mathbf{0}\]
   if $|\lambda-1|<1$ or $|\lambda_{1}+\lambda_{2}+......\lambda_{m}-1| < 1.$

\end{proof}

\begin{theorem}
\label{2}
let $\mathbf{T}$ be $NXN$ non negative, irreducible matrix then, $\exists $ a positive vector $\mathbf{t}$ such that \[(\mathbf{t})^{1+\alpha}=(\mathbf{D . T^t . t})\] $\mathbf{D}$ is diagonal matrix, with its $i^{th}$ element $d_i$ as $ \bigg[\frac{{\big(\mathbf{e_i C. diag(t) .t}}\big)^\alpha} {{\big(\mathbf{e_i C t}}\big)^{(1+\alpha)}}\bigg].$

\end{theorem}
\begin{proof}
Relation $(\mathbf{t})^{1+\alpha}=(\mathbf{D . T^t . t})$ can be written as 
\[\mathbf{T^tt}=\mathbf{D}^{-1}(\mathbf{t})^{1+\alpha}=\mathbf{y}\] where  $y_i=\bigg[\frac {{(\mathbf{e_i C t}})^{(1+\alpha)}}{{(\mathbf{e_i C. diag(t) .t}})^\alpha}\bigg]t_i^{1+\alpha}$.    Further, $y_i$ can be written as 
\[y_i=(\mathbf{e_i C t})\bigg[\frac {{(\mathbf{e_i C t}})^{\alpha}}{{(\mathbf{e_i C. diag(t) .t}})^\alpha}\bigg]t_i^{1+\alpha}\]

\[ =(\mathbf{e_i C t})\bigg[\frac {t_i(\mathbf{e_i C t})}{\mathbf{e_i C. diag(t) .t}}\bigg]^{\alpha}t_i\]
\[=(\mathbf{e_iCt}).f_i(t)\]
\[=\mathbf{e_i}(f_i(t)\mathbf{C}).\mathbf{t}\]
hence vector $\mathbf{y}$ can be written as \[\mathbf{y}=\mathbf{F(t)
.t}\] where matrix $\mathbf{F(t)}$ has nonzero elements at same position as matrix $\mathbf{C}$ and therefore at same position as  $\mathbf{T^t}$, its $ij$ element $F_{ij}(t)$ is $f_i(t)$, 
hence \[\mathbf{T^t.t}= \mathbf{F(t).t}\]   Now, $\exists $ a positive vector $\mathbf{t{'}}$, such that  $f_i(t^{'})\leq min(T_{ij}> 0)$. For such $\mathbf{t^{'}},$
\[\mathbf{T^t.t^{'}}> \mathbf{F(t^{'}).t^{'}}.\]

Also $\exists $ a positive vector $\mathbf{t{''}}$, such that $f_i(t^{''})\geq max(T_{ij})$, For such $\mathbf{t^{''}},$
\[\mathbf{T^t.t^{''}}< \mathbf{F(t^{''}).t^{''}}\] function $f$ is  continuous  and from  Lemma \ref{lemma1} there exist a path from $\mathbf{f(t^{'})}$ to $\mathbf{f(t^{''})}$ such that if  $\mathbf{t^{''}}> \mathbf{t^{'}}$, then $\mathbf{f(t^{''})}>\mathbf{f(t^{'})}.$ Hence, $\exists $ a positive vector $\mathbf{t}$ between $\mathbf{t^{'}}$ and $\mathbf{t^{''}}$, such that 
\[\mathbf{T^t.t}= \mathbf{F(t).t}\]
hence $\exists$ a positive vector $\mathbf{t}$, such that 
\[(\mathbf{t})^{1+\alpha}=(\mathbf{D . T^t . t})\]

\end{proof}

\begin{theorem}
\label{3}
Vector $\mathbf{t}$ in theorem \ref{2} is unique and can be calculated  by iterative function  
\[\mathbf{t^k}=\mathbf{\phi(t^{k-1})}=[\mathbf{D(t^{k-1}).T^t.t^{k-1})}]^\frac{1}{1+\alpha},\] where $\mathbf{t^k}$ is the value of vector $\mathbf{t}$ in $k^{th}$ iteration and $\mathbf{\phi}$ is the iterative function from $ \mathbb{R}^N\rightarrow \mathbb{R}^N $. The error in vector $\mathbf{t}$ will converge by the factor $\frac{1 + \alpha}{\alpha}$ in every iteration.

\end{theorem}

\begin{proof}
Let us rearrange the iterative function $\mathbf{\phi(t^{k-1})}$
\begin{equation}\label{equ9} 
\begin{split}
\mathbf{t^{k}} & = \mathbf{\phi(t^{k-1})} \\
     & = [\mathbf{D(t^{k-1}).T^t.t^{k-1})}]^\frac{1}{1+\alpha}\\
     & = [\mathbf{diag(d_1,d_2....d_N) T^t. t^{k-1}}]^\frac{1}{1+\alpha}    
\end{split}
\end{equation}

where \[d_i = \Bigg(\frac{{\big(\mathbf{e_i C. diag(t^{k-1}) .t^{k-1}}}\big)^\alpha} {{\big(\mathbf{e_i C t^{k-1}}}\big)^{(1+\alpha)}} \Bigg)\] 
 Then $i^{th}$ element of $\mathbf{t^k}$ will be 
 
\begin{equation}\label{equ10}
\begin{split}
 t_i^{k} & =\Bigg[\frac{\mathbf{\big(e_i T^t.t^{k-1}\big)\big(e_i C. diag(t^{k-1}).t^{k-1}\big)}^\alpha}{\big(\mathbf{e_i C t^{k-1}}\big)^{1+\alpha}}\Bigg]^\frac{1}{1+\alpha}\\ 
  & = \Bigg[\frac{\mathbf{\big(e_i T^t.t^{k-1}\big)}^\frac{1}{1+\alpha}\mathbf{\big(e_i C. diag(t^{k-1}).t^{k-1}\big)}^\frac{\alpha}{1+\alpha}}{\mathbf{\big(e_i C t^{k-1}\big)}}\Bigg]
\end{split}
\end{equation}

Let $t_i^k$ and $t_i^{k-1}$ are, far from actual solution $t_i$ by  $\delta t_i^k$ and $\delta t_i^{k-1}$ respectively, then  
\begin{multline*}
t_i + \delta t_i^{k}=  \Bigg[\frac{\mathbf{\big(e_i T^t.(t + \delta t^{k-1})\big)}^\frac{1}{1+\alpha}}{\mathbf{\big(e_i C (t + \delta t^{k-1})\big)}}\Bigg].\\
[\mathbf{\big(e_i C. diag(t + \delta t^{k-1}).(t + \delta t^{k-1})\big)}^\frac{\alpha}{1+\alpha}]
\end{multline*}

\begin{multline*}
=\Bigg[\frac{\mathbf{\big(e_i T^t.t\big)}^\frac{1}{1+\alpha}\mathbf{\big(e_i C. diag(t).t\big)}^\frac{\alpha}{1+\alpha}}{\mathbf{\big(e_i C t\big)}}\Bigg].\Bigg[\frac{(1+\frac{\mathbf{e_i.T^t.\delta t^{k-1}}}{\mathbf{e_i.T^t.t}})^\frac{1}{1+\alpha}}{(1+\mathbf{\frac{e_i.C.\delta t^{k-1}}{e_i.C.t})}}\Bigg].\\\Bigg[\bigg(1+{\frac{2\mathbf{e_i.C. diag(t).\delta t^{k-1}}}{{\mathbf{e_i.C. diag(t).t}}}}+\mathbf{\frac{e_i.C. diag(\delta t^{k-1}).\delta t^{k-1}}{{e_i.C. diag(t).t}}\bigg)}^\frac{\alpha}{1+\alpha}\Bigg]
\end{multline*}
Since $\mathbf{\delta t^{k-1}} << \mathbf{t}$ hence we can neglect the higher order terms of $\mathbf{\delta t^{k-1}}$.
\begin{multline*}
\approx \Bigg[\frac{\mathbf{\big(e_i T^t.t\big)}^\frac{1}{1+\alpha}\mathbf{\big(e_i C. diag(t).t\big)}^\frac{\alpha}{1+\alpha}}{\mathbf{\big(e_i C t\big)}}\Bigg].\Bigg[\frac{(1+\mathbf{\frac{e_i.T^t.\delta t^{k-1}}{e_i.T^t.t})}^\frac{1}{1+\alpha}}{(1+\mathbf{\frac{e_i.C.\delta t^{k-1}}{e_i.C.t})}}\Bigg].\\\Bigg[\bigg(1+\frac{2\mathbf{e_i.C. diag(t).\delta t^{k-1}}}{\mathbf{e_i.C. diag(t).t}}\bigg)^\frac{\alpha}{1+\alpha}\Bigg]
\end{multline*}

Using equation \ref{equ10} 

\begin{multline*}
=t_i. \Bigg[\frac{(1+\mathbf{\frac{e_i.T^t.\delta t^{k-1}}{e_i.T^t.t}})^\frac{1}{1+\alpha}}{(1+\mathbf{\frac{e_i.C.\delta t^{k-1}}{e_i.C.t})}}\Bigg].\\\Bigg[\bigg(1+\frac{2\mathbf{e_i.C. diag(t).\delta t^{k-1}}}{{\mathbf{e_i.C. diag(t).t}}}\bigg)^\frac{\alpha}{1+\alpha}\Bigg]
\end{multline*}
Using binomial expansion and neglecting higher order terms of $\mathbf{\delta  t^{k-1}}$

\begin{multline*}
\approx t_i.\Bigg[\frac{(1+\frac{\mathbf{e_i.T^t.\delta t^{k-1}}}{{(1+\alpha)}\mathbf{e_i.T^t.t}})}{(1+\frac{\mathbf{e_i.C.\delta t^{k-1}}}{\mathbf{e_i.C.t}})}\Bigg].\\\Bigg[\bigg(1+\frac{2\alpha \mathbf{e_i.C. diag(t).\delta t^{k-1}}}{{(1+\alpha)\mathbf{e_i.C. diag(t).t}}}\bigg)\Bigg]
\end{multline*}

\begin{multline*}\
\delta t_i^{k}=t_i.\\\Bigg[\frac{\big(1+\frac{\mathbf{e_i.T^t.\delta t^{k-1}}}{{(1+\alpha)}\mathbf{e_i.T^t.t}}\big)}{\big(1+\frac{\mathbf{e_i.C.\delta t^{k-1}}}{\mathbf{e_i.C.t}}\big)}\bigg(1+\frac{2\alpha \mathbf{e_i.C. diag(t).\delta t^{k-1}}}{(1+\alpha)\mathbf{e_i.C. diag(t).t}}\bigg)\Bigg] - t_i
\end{multline*}

\begin{multline*}\
\delta t_i^{k}=t_i.\\\Bigg[\frac{\big(1+\frac{\mathbf{e_i.T^t.\delta t^{k-1}}}{{(1+\alpha)}\mathbf{e_i.T^t.t}}\big)}{(1+\frac{\mathbf{e_i.C.\delta t^{k-1}}}{\mathbf{e_i.C.t}})}\bigg(1+\frac{2\alpha \mathbf{e_i.C. diag(t).\delta t^{k-1}}}{{(1+\alpha)\mathbf{e_i.C. diag(t).t}}}\bigg) - 1\Bigg]
\end{multline*}

\begin{multline*}\
\hspace{5 mm} = t_i.\\\Bigg[\Bigg(1+\frac{\mathbf{e_i.T^t.\delta t^{k-1}}}{{(1+\alpha)}\mathbf{e_i.T^t.t}}\Bigg)\Bigg(1+\frac{2\alpha \mathbf{e_i.C. diag(t).\delta t^{k-1}}}{{(1+\alpha)\mathbf{e_i.C. diag(t).t}}}\Bigg)\\ - \Bigg(1+\frac{\mathbf{e_i.C.\delta t^{k-1}}}{\mathbf{e_i.C.t}}\Bigg)\Bigg]\Bigg/\Bigg(1+\frac{\mathbf{e_i.C.\delta t^{k-1}}}{\mathbf{e_i.C.t}}\Bigg)
\end{multline*}
Approximating the denominator term $\bigg(1+\frac{\mathbf{e_i.C.\delta t^{k-1}}}{\mathbf{e_i.C.t}}\bigg)\approx 1$

\begin{multline*}\
\hspace{5 mm}\approx t_i.\\\Bigg[\Bigg(\frac{\mathbf{e_i.T^t.\delta t^{k-1}}}{{(1+\alpha)}\mathbf{e_i.T^t.t}}\Bigg)+\Bigg(\frac{2\alpha \mathbf{e_i.C. diag(t).\delta t^{k-1}}}{{(1+\alpha)\mathbf{e_i.C. diag(t).t}}}\Bigg)\\+\Bigg(\frac{\mathbf{e_i.T^t.\delta t^{k-1}}}{{(1+\alpha)}\mathbf{e_i.T^t.t}}\Bigg).\Bigg(\frac{2\alpha \mathbf{e_i.C. diag(t).\delta t^{k-1}}}{{(1+\alpha)\mathbf{e_i.C. diag(t).t}}}\Bigg)\\ - \Bigg(\frac{\mathbf{e_i.C.\delta t^{k-1}}}{\mathbf{e_i.C.t}}\Bigg)\Bigg]
\end{multline*}
Again neglecting higher order terms of $\mathbf{\delta t^{k-1}}$

\begin{multline*}\
\hspace{5 mm}\approx t_i.\\\Bigg[\Bigg(\frac{\mathbf{e_i.T^t.\delta t^{k-1}}}{{(1+\alpha)}\mathbf{e_i.T^t.t}}\Bigg)+\Bigg(\frac{2\alpha \mathbf{e_i.C. diag(t).\delta t^{k-1}}}{{(1+\alpha)\mathbf{e_i.C. diag(t).t}}}\Bigg)\\ - \Bigg(\frac{\mathbf{e_i.C.\delta t^{k-1}}}{\mathbf{e_i.C.t}}\Bigg)\Bigg]
\end{multline*}

\begin{multline*}\
\hspace{5 mm} = \\\Bigg[\Bigg(\frac{t_i.\mathbf{ e_i.T^t}}{{(1+\alpha)}\mathbf{e_i.T^t.t}}\Bigg)+\Bigg(\frac{2\alpha .t_i. \mathbf{e_i.C. diag(t)}}{{(1+\alpha)\mathbf{e_i.C .diag(t).t}}}\Bigg)\\ - \Bigg(\frac{t_i. \mathbf{e_i.C}}{\mathbf{e_i.C.t}}\Bigg)\Bigg].\mathbf{\delta t^{k-1}}
\end{multline*}

\begin{multline*}\
= [\mathbf{X_i}+\mathbf{Y_i}-\mathbf{Z_i}].\mathbf{\delta t^{k-1}}
\end{multline*}
Where $\mathbf{X_i, Y_i}$ and $\mathbf{Z_i}$ are $i^{th}$ row of $NXN$ matrix $\mathbf{X, Y}$ and $\mathbf{Z}$ respectively, it is clear from the  above that
\[ \mathbf{Xt}=\frac{1}{1+\alpha} . \mathbf{t}, \]
\[ \mathbf{Yt}=\frac{2\alpha}{1+\alpha} . \mathbf{t}\] and \[\mathbf{Zt}= 1 . \mathbf{t},\] where,\par \[ \mathbf{X, Y, Z, t} > \mathbf{0} \] Matrices $ \mathbf{X,Y,Z} $ have non zero elements  at same position as matrix $\mathbf{T^t}$, hence all these  are also irreducible. Therefore spectral radius of $\mathbf{X , Y}$ and $\mathbf{Z}$ will be $\frac{1}{1+\alpha}$, $\frac{2\alpha}{1+\alpha}$ and 1. \\Now \[\mathbf{\delta t^k} =(\mathbf{X+Y-Z)\delta t^{k-1}}\] 
If initial error in $\mathbf{t}$ is $\mathbf{\delta t^0}$ then

\[lim_{k\to\infty}\mathbf{\delta t^k} =lim_{k\to\infty}(\mathbf{X+Y-Z})^k\mathbf{\delta t^0} \]
 Directly from Theorem \ref{0}\\
\[ lim_{k\to\infty}(\mathbf{X+Y-Z})^k\mathbf{\delta t^0} = \mathbf{0}\]
 
\[ \Rightarrow lim_{k\to\infty}\mathbf{\delta t^k} =\mathbf{0}\]
if \[|\frac{1}{1+\alpha} + \frac{2\alpha}{1+\alpha} -1| < 1 \]
\[\Rightarrow \frac{\alpha}{1+\alpha} < 1\]
Which is true for any $ \alpha > 0 $ \\

   In every step error will decrease by a factor of $ \frac{1+\alpha}{\alpha}$ . Vector $\mathbf{t}$ will converge fast if value of $\alpha$ is small. Hence  speed of convergence will be 
\[=\frac{1 + \alpha}{\alpha}\]
    
\end{proof}

In theorem \ref{3} if $\alpha=q/p$ then $i^{th}$ element of vector $\mathbf{t}$ will be  \[t_i=\Bigg(\frac{ \big(\mathbf{e_iC. diag(t).t}\big)^\frac{q}{p}}{(\mathbf{e_iCt})^{(1+\frac{q}{p})}} \mathbf{(e_iT^t t)}\Bigg)^\frac{1}{1+\frac{q}{p}}\] 

\[=\Bigg(\bigg(\frac{ \mathbf{e_iC. diag(t).t}}{\mathbf{e_iCt}}\bigg)^q \bigg(\frac{\mathbf{e_iT^t t}}{\mathbf{e_iCt}}\bigg)^p\Bigg)^\frac{1}{p+q}\]

\[t_i=\Bigg[\Bigg(\frac{\sum_{j\in S}{T_{ji}t_j}}{\sum_{j\in S}{t_j}}\Bigg)^p . \Bigg(\frac{\sum_{j\in S}{t_{j}^2}}{\sum_{j\in S}{t_j}}\Bigg)^q\Bigg]^\frac{1}{(p+q)} \]  Which is equation \ref{equ4} hence global trust exists, and can be calculated by above equation.
\section{Analysis of Algorithm} 
  \subsection{Speed of Convergence} 
  Figure \ref{converge} shows the speed of convergence of our algorithm, graph is plotted for the average value of residue of global trust at any node as the iterations performs, i.e.\[\frac{1}{N}||\mathbf{t^k} - \mathbf{t^{k-1}}||_1\] We can see  from the figure that as $\alpha$ is decreasing its speed of convergence is increasing. For $\alpha \leq 1/3$,  it is converging in less then seven iterations. Lower value of $\alpha$ means higher value of $p$ compare to $q$. Higher $p$ means more weightage to first term, which is weighted average of trust value informed one hop neighbors. Lower $q$ means lesser weight to second term which is equivalent global trust of trust assigning set $S$.  But there is trade-off between these two. First term is used to settle the conflicts among direct trust assigning peers, and second term is used to  bias the global trust of peer according to global trust of the trust assigning set. The global trust of individual members of set is biased by their trust assigning set respectively, and so on. Hence second factor is taking the opinion  from the whole of the network. We can not neglect the opinion of other peers but we also need faster convergence of the algorithm. Because higher  the speed of convergence, lesser will be the number of message needed to update the global trust.

%Figure for residual start from here

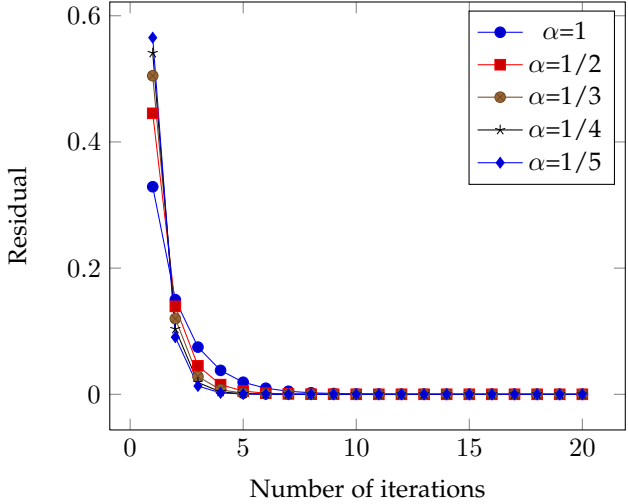
\begin{figure}
  \begin{tikzpicture}
  \begin{axis}[x tick label style={
		/pgf/number format/1000 sep=}, 
    xlabel={Number of iterations},
    ylabel={Residual},
    %xmin=0, xmax=20,
    %ymin=0, ymax=.2,
    %xtick={0,2,4,6,8,10,12,14,16,18,20},
    %ytick={0,.020,.04,.06,.08,.1,.12,.14,.16,.18,.20,.22},
    %legend pos=north east,
    %ymajorgrids=true,
    %grid style=dashed,
]
 
%\addplot
    %color=blue,
   % mark=square,
    \addplot
    coordinates {(	1	,	0.329	)
(	2	,	0.1499	)
(	3	,	0.0748	)
(	4	,	0.0379	)
(	5	,	0.0191	)
(	6	,	0.0096	)
(	7	,	0.0048	)
(	8	,	0.0024	)
(	9	,	0.0012	)
(	10	,	0.0006	)
(	11	,	0.0003	)
(	12	,	0.0001	)
(	13	,	0.0001	)
(	14	,	0	)
(	15	,	0	)
(	16	,	0	)
(	17	,	0	)
(	18	,	0	)
(	19	,	0	)
(	20	,	0	)
    };

 %\addplot
    %color=green,
    %mark=circle,
    \addplot
    coordinates {
   (	1	,	0.4454	)
(	2	,	0.1393	)
(	3	,	0.0451	)
(	4	,	0.0153	)
(	5	,	0.0051	)
(	6	,	0.0017	)
(	7	,	0.0006	)
(	8	,	0.0002	)
(	9	,	0.0001	)
(	10	,	0	)
(	11	,	0	)
(	12	,	0	)
(	13	,	0	)
(	14	,	0	)
(	15	,	0	)
(	16	,	0	)
(	17	,	0	)
(	18	,	0	)
(	19	,	0	)
(	20	,	0	)

    };
     \addplot
    coordinates {
     (	1	,	0.5049	)
(	2	,	0.1199	)
(	3	,	0.0281	)
(	4	,	0.0072	)
(	5	,	0.0018	)
(	6	,	0.0004	)
(	7	,	0.0001	)
(	8	,	0	)
(	9	,	0	)
(	10	,	0	)
(	11	,	0	)
(	12	,	0	)
(	13	,	0	)
(	14	,	0	)
(	15	,	0	)
(	16	,	0	)
(	17	,	0	)
(	18	,	0	)
(	19	,	0	)
(	20	,	0	)
};
\addplot
    coordinates {(	1	,	0.541	)
(	2	,	0.1035	)
(	3	,	0.0187	)
(	4	,	0.0038	)
(	5	,	0.0008	)
(	6	,	0.0002	)
(	7	,	0	)
(	8	,	0	)
(	9	,	0	)
(	10	,	0	)
(	11	,	0	)
(	12	,	0	)
(	13	,	0	)
(	14	,	0	)
(	15	,	0	)
(	16	,	0	)
(	17	,	0	)
(	18	,	0	)
(	19	,	0	)
(	20	,	0	)
};
\addplot
    coordinates{
(	1	,	0.5653	)
(	2	,	0.0905	)
(	3	,	0.0131	)
(	4	,	0.0022	)
(	5	,	0.0004	)
(	6	,	0.0001	)
(	7	,	0	)
(	8	,	0	)
(	9	,	0	)
(	10	,	0	)
(	11	,	0	)
(	12	,	0	)
(	13	,	0	)
(	14	,	0	)
(	15	,	0	)
(	16	,	0	)
(	17	,	0	)
(	18	,	0	)
(	19	,	0	)
(	20	,	0	)

    };
\legend{$\alpha$=1,$\alpha$=1/2,$\alpha$=1/3,$\alpha$=1/4,$\alpha$=1/5,}
\end{axis}
\end{tikzpicture}
\centering
\caption{Convergence of Algorithm for different values of $\alpha$}\label{converge}
\end{figure}

 \subsection{Implementation in Distributed System}
  Algorithm \ref{algo1} describes how the requesting peers can select the  peer from whom to download. We will call the selected peers as source peers. Each peer can set the $Global\_ref$, a reference value of global trust, to decide whether to select a peer as a  source peer or not. If global trust of any peer is less than $Global\_ref$, then it should not be selected as a  source. The requesting peers initiates a query for a resources. Each query is given a TTL value. Whenever a query is forwarded, its TTL value is decremented. When TTL becomes zero, the query is not forwarded anymore. The requesting peer can control the scope of query by choosing TTL value. A requesting peer will wait for a time greater than 2xTTL. If no response is received within waiting period, the query can be  made again with larger TTL value. After getting the response from the network, a peer can select the most reputed peer as the  source peer and can download the required file. In order to balance the load of the network, a peer can select the set of peers whose global trust is  more than the $Global\_ref$ and then the source peer can be selected  probabilistically among them. The probability of selecting any peer as a  source can be taken to be proportional to its global trust. This strategy has twofold effects, one is to allow only the reputable peer to become the  source,  and another  balancing the load among the reputable peers. After selecting the  source peer and getting the  file from it, a peer can evaluate  quality of file and can send the  feedback to the peers  holding the trust values of source peer. Here it is important to note  that if local trust of any source peer is updated then it will not effect the local trust of other co-source peers, and peer needs to send the updated feedback of only that  source peer whose local trust is being updated. However in normalization methods \cite{eigen}\cite{power}, if local trust value of any source peer is updated then local trust value of all other source peers also have to be updated. Because in normalization method, sum of all the trust values assigned by any peer to all its source peers have to be one. This guarantees the convergence of global trust in normalization method.   Let average number of source peers per peer is $avg\_source$ then in one update of local trust we are saving $avg\_source - 1 $ number of messages, which is very significant for the whole network.  In this process, if all of the responding peers have a global trust less than $Global\_ref$ then all of them can be rejected and  requesting peer can go for another search by increasing the TTL value of query. There should be an upper limit on TTL, after which peer should stop and terminate the query process.\par Global trust can be updated by each trust holder peer according to algorithm \ref{algo2}.  It is similar to update method used in  \cite{eigen}\cite{power}. To calculate the global trust value of any peer, trust holder peer needs to know the local trust values of that peer and the current global trust value of trust assigning peers. Trust assigning peers will send the local trust values of source peers to their trust holder peers and trust holder peers will ask the current global trust values of trust assigning peers  from their respective trust holder peers. This process is  repeated till the convergence of global trust see algorithm \ref{algo2}. For security purpose, more than one peers can manage the global trust of a particular peer. Again here, it is important to note that number of iterations required to converge the global trust will be maximum in first time only. In all successive updates global trust will converge more faster since initial guess of global trust will be more close to the final global trust value. Global trust will converge for any initial value of global trust vector, $\mathbf{t^0}$. But we have to ensure that, at least one component of global trust vector, $\mathbf{t^0}$, must be along the final global trust vector, $\mathbf{t}$, see section 4. We can ensure this by taking initial global trust value for all peers as $(w\_g + w\_b)/2$.\par So far we have discussed that how to aggregate the global trust from local trust and in peer selection procedure we are considering only the global trust. However global trust is more significant if peer has no past history with any of responding peer. If peer has some past history with any one of them then decision of selection of source peer can be done according to $\beta t_i + (1-\beta) T_{ji} $. It is  convex combination of global trust of peer and local trust value assigned by requesting peer to responding peer in past. The value of parameter $\beta$ can be selected by peer depending on its confidence on responding peer.


\begin{algorithm}
\caption{For selection of source peer}\label{algo1}
\begin{algorithmic}[1]
\Procedure{}{}
\State $Global\_ref \gets \frac{(w\_g+w\_b)}{2} $
\State $ TTL  \gets Const.$
\BState \emph{top}:
\State Set $Time\_Counter \geq 2 * TTL$
\State $ i \gets 0 $
\State Send the query for required file in Network;

\While{$i \leq Time\_Counter $ }
  \State Wait for response from the network;
   \State $i\gets i + 1 $;
\EndWhile
\If{$Number\_of\_responding\_peers == 0$}
\If{$TTL \geq (TTL)_{upper}$}
\State Terminate the query process;
\Else 
\State Increase $TTL$;
\State \textbf{goto} top
\EndIf
\Else
\State Get the $Global\_Trust$ of all the responding peers from their trust holder peer;
\State Select the peer with maximum $Global\_Trust$;
\If{$Global\_Trust \geq Global\_ref $}
\State  Download the required file;
\State  Evaluate the file;
\State  Send the feedback to trust holder peer of source peer;
\State Stop; 
\Else
\If{$TTL \geq (TTL)_{upper}$}
\State Terminate the query process;
\Else 
\State Increase $TTL$;
\State \textbf{goto} top
%\State \textbf{close};
\EndIf
\EndIf
\EndIf
\EndProcedure
\end{algorithmic}
\end{algorithm}

\begin{algorithm}
\caption{For updating the Global Trust of peers}\label{algo2}
\begin{algorithmic}[1]
\State \textbf{Input:} Local Trust values of peers
\State \textbf{Output:} Global trust on trust holder peers
\Procedure{}{}
\For { each peer $i$ }
\ForAll {peer $j$, who is selected as source peer}
\State Evaluate the received file; 
\State Assign the Local Trust value between $w\_b$ to $w\_g$;
\State Send the local Trust to trust holder peer of peer $j$;
\EndFor
\If { Peer $i$ is trust holder peer of peer $k$}

\ForAll{ peer $j$, who  selected $k$ as a source peer }
\State Receive the Local Trust values $T_{kj}$; 
\State Locate their trust holder peer; 
\EndFor
\State Initialization;
\State Set $p$,$q$, $previous\_t_k$, $threshold$;
\While{$ error \geq threshold$}
\State Receive the Global Trust $t_j$ from their trust holder peer ;
\State Compute
\State $t_k \gets \Bigg[\Bigg(\frac{\sum_{j\in S}{T_{jk}t_j}}{\sum_{j\in S}{t_j}}\Bigg)^p . \Bigg(\frac{\sum_{j\in S}{t_{j}^2}}{\sum_{j\in S}{t_j}}\Bigg)^q\Bigg]^\frac{1}{(p+q)}$
\State  $error \gets |t_k - previous\_t_k|$
\State  $ previous\_t_k \gets t_k $
\EndWhile
\EndIf
\EndFor

\EndProcedure
\end{algorithmic}
\end{algorithm}

\section{Experimental Evaluation}

%graph for Good node gives value to good file			
%Bad node gives value to bad file
Referring to \cite{sfa} and \cite{eval}, we used NetLogo 5.2\cite{netlogo}, to evaluate the performance of our algorithm. NetLogo is a multi-agent programmable modeling environment where we can model the different agents and can ask them to perform the task in parallel and independently. It is written mostly in Scala, with some parts in Java. We also simulated  and compared our result with two most popular reputation system Eigen Trust\cite{eigen} and Power Trust \cite{power}. We found that our algorithm is giving better performance in various behavioral conditions of peers in the network. It is  explained  in next  subsections.
\subsection{Simulation setup}	

\begin{table}
\begin{center}
\caption{Values of various parameters which we used in our simulation }\label{table2}
\begin{tabular}{ | m{.25cm} | m{1.75cm} | m{4cm}  | m{0.5cm}|} 
 \hline
    S.N.& Parameter & Description & Value\\
   \hline
  1 &$N$ & Number of Peers in Network&100  \\ 
 \hline
  2 &$Num\_file$ & Number of different files in Network & 1000  \\
   \hline
 3&$Num\_transact$ & Total  number of transaction in network &10000 \\
  \hline
 4& $\gamma$ &Zipf's Constant & 0.4  \\ 
  \hline
 5& $p$ &See from equation \ref{equ4} &3 \\
 \hline
 6& $q$& See from equation \ref{equ4} &1 \\
   \hline
  7& $w\_b$ &Weight factor for unsatisfactory file& 1 \\
  \hline
  8&$w\_g$&Weight factor for satisfactory file  & 10 \\
  \hline
  9 & $Global\_ref$ & Threshold value of Global Trust for good peers& 5.5\\
  \hline
   
\end{tabular}
\end{center}
\end{table}
	
 We simulate a typical p2p network with parameters and distribution taken from real world measurements\cite{netmodel}. We used percentage of authentic download as a standard metric to evaluate and compare the performance of reputation systems. In this model, peer can issue a query for particular file. The query propagates in the network and peers can respond to it if they have that particular  file. Peers can ask for only those files which they don't have. Peers can select a source peer according to its global trust and can reject all the possible source peer if none of them are found suitable. \par The parameters used in simulation are shown in Table \ref{table2}. We have taken 100 node network. However number of nodes can be increased upto any number but the results are expected to  remain same irrespective of the number of nodes, because we are taking percentage of authentic download as a standard metric to compare the results. Files are distributed among the nodes as a Zipf's Law with Zipf constant as 0.4. There are 1000 types of different files in network.  We have taken the weight factor for authentic file is 10 and  for malicious file is 1. Observations are taken for 10000 transactions. Rest of the parameters are shown in the  Table \ref{table2}. We have also considered the transient phase and applied the global trust update after every 200 query cycles.
\par Peer behaves as their defined behavior in 95\% of the time and rest of the  time just opposite to it, assuming that 5\% of time all peers make the mistake and behave just opposite to their defined behavior.  There are good peers, malicious peers, unpredictable peers, and collective peers in the network. Behavior of these peers are defined in next subsections.  We have performed simulation in various peers behavioral conditions. Each experiment is performed for ten times and then readings are averaged among them.
\subsection{Performance in the presence of pure malicious peers}
In this model, there are some good and some malicious peers in network. Good peers provide authentic files and right feedback and malicious peers provide inauthentic file and wrong feedback. 95\% of the times peer behaves as above defined behavior and rest of the time just opposite to it. Wrong feedback can be given in many ways \cite{sort}, however we are considering the case in which, malicious peers give highest feedback to other malicious peers and lowest feedback to good peers.  We have increased the percent of  malicious peers in the network from 5\% to 45\% and plotted the result in figure \ref{malicious}. It can be  observed from this figure that percentage of authentic download increase significantly as compare to the Eigen Trust\cite{eigen} and Power Trust\cite{power}. Here we can see that maximum percent of authentic download is 95\% because 5 \% of time peers behave just opposite to their defined behavior.
 
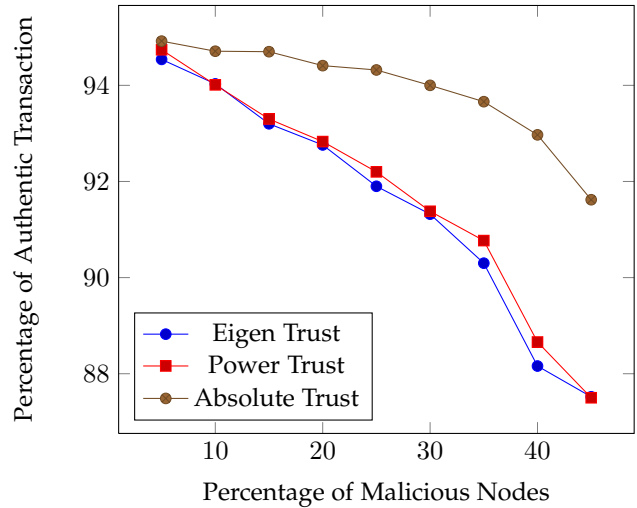
\begin{figure}
\begin{tikzpicture}
\begin{axis}[
	x tick label style={
		/pgf/number format/1000 sep=},
	ylabel= Percentage of Authentic Transaction ,
	xlabel= Percentage of Malicious Nodes,
	legend pos=south west,
	%ybar interval=1,
]
\addplot 
	coordinates {(5	,94.54)
(10,	94.03)
(15,	93.20)
(20,	92.76)
(25,	91.9)
(30,	91.32)
(35	,	90.30)
(40,	88.16)
(45	,	87.52)

};
\addplot 
	coordinates {(5		,94.74)(10	,	94.01)(15	,	93.3)(20,		92.83)(25	,	92.2)(30	,	91.38)(35	,	90.77)(40	,	88.66)
(45	,	87.5)
};
\addplot
coordinates {(5,	94.92)
(10,	94.71)
(15,    94.7)
(20,	94.41)
(25,	94.32)
(30,	94.00)
(35,	93.66)
(40,	92.97)
(45,	91.62)
};
%\addplot
%coordinates {(5,	90.65)
%(10,	86.20)
%(15,    81.31)
%(20,	76.85)
%(25,	72.99)
%(30,	67.84)
%(35,	63.52)
%(40,	59.51)
%(45,	54.53)
%};
\legend{Eigen Trust, Power Trust,Absolute Trust}
\end{axis}
\end{tikzpicture}
\caption{good peers gives value to good peers and bad peers give value to bad peers}\label{malicious}
\end{figure}
\subsection{Performance in the presence of peers with unpredictable behavior }
%10 % of nodes are purely malicious few nodes behave good upto %200 query cycles and then behave maliciously	
In this model, simulation is performed in the presence of unpredictable peers. These peers behave as a good peers upto some time and earn the reputation, then they start behaving like malicious peers after some time. Since behavior of these peers change dynamically and it is very difficult to predict their behavior, hence they are called unpredictable peers. Initially there are 10\% of purely malicious peers and few unpredictable peers in the network.  We increase the percentage of unpredictable peers from 5\% to 35\% and plotted the result in figure \ref{unpredictable}. We can see from this figure that Absolute Trust perform significantly better. We can also see that power trust is vulnerable to unpredictable malicious peers attack. It is because in some cases unpredictable malicious peers can earn the good reputation upto some time and can be elected as a power node in the network, and then they can start misusing this reputation. Since power nodes play a major role in Power Trust\cite{power} so they can damage the system more powerfully, if wrong nodes are elected as a power nodes. Chances of, unpredictable peers getting elected as a power node is higher when  percentage of malicious nodes are more. This we can observe from figure \ref{unpredictable} that, when \% of unpredictable peers increase upto 25\% the authentic download decreases rapidly. 					
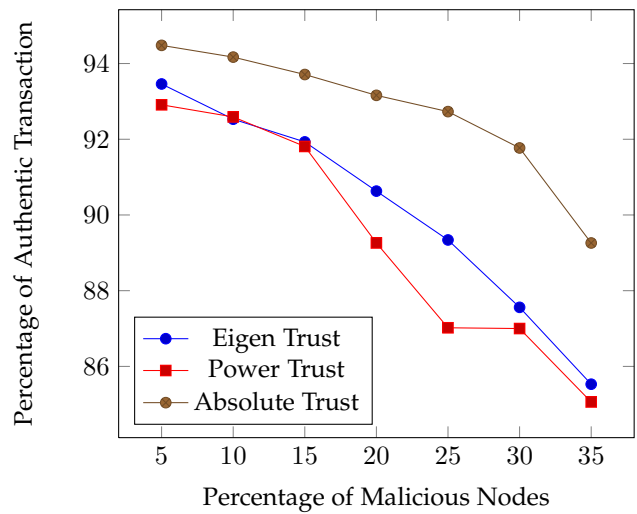
\begin{figure}
\begin{tikzpicture}
\begin{axis}[
	x tick label style={
		/pgf/number format/1000 sep=},
	ylabel= Percentage of Authentic Transaction ,
	xlabel= Percentage of Malicious Nodes,
	legend pos=south west,
	%ybar interval=1,
]
\addplot 
	coordinates {(	5	,	93.46	)
(	10	,	92.53	)
(	15	,	91.93	)
(	20	,	90.63	)
(	25	,	89.34	)
(	30	,	87.56	)
(	35	,	85.53	)

};
\addplot 
	coordinates {(	5	,	92.91	)
(	10	,	92.59	)
(	15	,	91.81	)
(	20	,	89.26	)
(	25	,	87.02	)
(	30	,	87.0	)
(	35	,	85.06	)

};
\addplot
coordinates {(	5	,	94.48	)
(	10	,	94.17	)
(	15	,	93.71	)
(	20	,	93.16	)
(	25	,	92.73	)
(	30	,	91.77	)
(	35	,	89.26	)

};
\legend{Eigen Trust, Power Trust,Absolute Trust}
\end{axis}
\end{tikzpicture}
\caption{10 \% of peers are purely malicious and few behave as a good peers upto some time then behave maliciously after that}\label{unpredictable}
\end{figure}
\subsection{Performance in the presence of malicious collectives}
%malicious Collectives
In this model, we performed the simulation in the presence of malicious collectives. Malicious collectives are group of peers, those who know each others and increase their reputation values and give minimum values to all others. Malicious collectives always select the source peers from their  group and increase their reputation by giving maximum weight to their file and minimum to all others file. If there are more than one malicious groups then malicious collectives select the source peer among their own group and for other group they behave like pure malicious peer. Practically speaking, percent of peers making malicious collective can not be more than 5\% to 10\%  however there can be many number of malicious groups. Keeping this thing in view we kept 5\% of the peers in one group and number of groups are increased from 1 to 6. Result of simulation is plotted in figure \ref{collective}. We can observe that  Absolute trust is performing better than rest two.
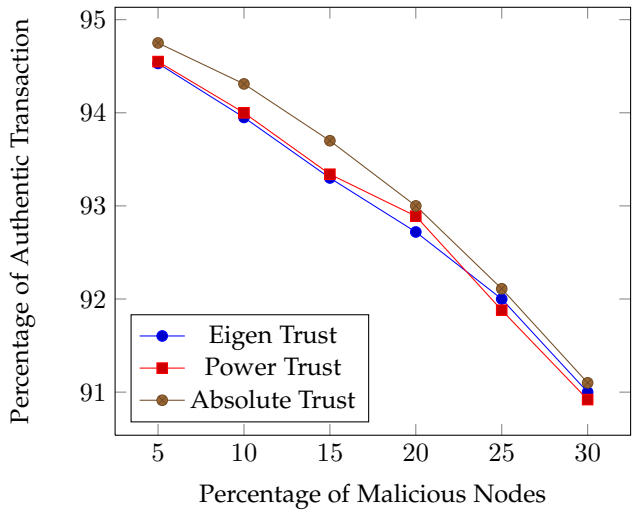
\begin{figure}
\begin{tikzpicture}
\begin{axis}[
	x tick label style={
		/pgf/number format/1000 sep=},
	ylabel= Percentage of Authentic Transaction ,
	xlabel= Percentage of Malicious Nodes,
	legend pos=south west,
	%ybar interval=1,
]
\addplot 
	coordinates {(	5	,	94.53	)
(	10	,	93.95	)
(	15	,	93.3	)
(	20	,	92.72	)
(	25	,	92	)
(	30	,	91	)

};
\addplot 
	coordinates {(	5	,	94.55	)
(	10	,	94	)
(	15	,	93.34	)
(	20	,	92.89	)
(	25	,	91.88	)
(	30	,	90.92	)

};

\addplot 
	coordinates {(	5	,	94.75	)
(	10	,	94.31	)
(	15	,	93.7	)
(	20	,	93	)
(	25	,	92.11	)
(	30	,	91.1	)

};
\legend{Eigen Trust, Power Trust,Absolute Trust}
\end{axis}
\end{tikzpicture}
\caption{Malicious collectives}\label{collective}
\end{figure}
\subsection{Analysis of load distribution among the peers}
Among all the responding peers, a peer is selected as a source peer if its global trust is higher as compare to others. In the Absolute trust we are calculating the global trust of peers absolutely so it is not reducing the global trust of others, while in relative ranking peers are combating with each other for global trust, if global trust of any peer is increased in some fraction it will decrease the global trust of other peers. Hence relative difference between the global trust, will always be higher in case of normalization  .
\par In simulation we calculated the load of individual peers that is number of times a particular peer is selected as a source peer. Then we calculated the standard deviation of load among all the peers. Simulation is performed in the presence of pure malicious peer, and standard deviation of load is  calculated only among the good nodes   because majority of load have to be shared by good peers only. Malicious peers are increased from 5\% to 30\%. The result of simulation is plotted in Figure \ref{standev}, we can see directly from there, the standard deviation of load distribution is minimum in Absolute trust.
\begin{figure}
\begin{tikzpicture}
\begin{axis}[
	x tick label style={
		/pgf/number format/1000 sep=},
	ylabel= Standard Deviation of Load ,
	xlabel= Percentage of Malicious Nodes,
	legend pos=south east,
	%ybar interval=1,
]
\addplot 
	coordinates {(	5	,	149	)
(	10	,	149.88	)
(	15	,	150.08	)
(	20	,	154.4	)
(	25	,	156.36	)
(	30	,	158.16	)

};
\addplot 
	coordinates {(	5	,	149.04	)
(	10	,	151.44	)
(	15	,	154.02	)
(	20	,	155.109	)
(	25	,	157.9	)
(	30	,	160.21	)

};
\addplot
coordinates {(	5	,	118.107	)
(	10	,	119.76	)
(	15	,	133.21	)
(	20	,	138.31	)
(	25	,	144.9	)
(	30	,	152.93	)

};
\legend{Eigen Trust, Power Trust,Absolute Trust}
\end{axis}
\end{tikzpicture}
\caption{Standard deviation of load distribution among the peers in different Reputation System}\label{standev}
\end{figure}
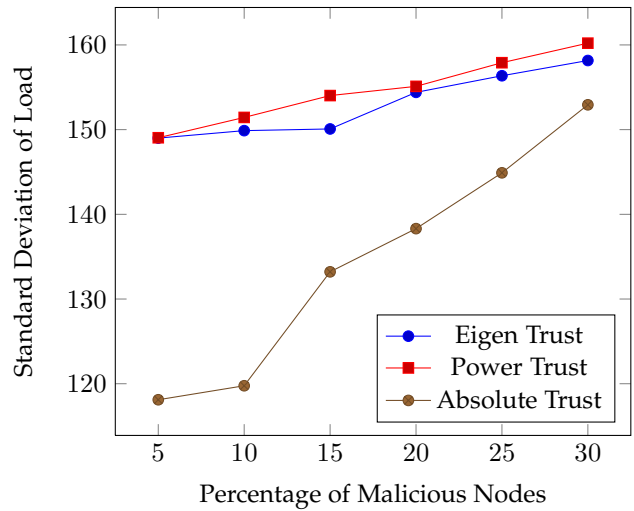

\section{Conclusion and Future work}
In this work, we presented an algorithm for aggregation of local trust  in peer-to peer network. We have seen that our algorithm is able to fulfill all design considerations mentioned in the introduction.  Aggregation is done without normalization, hence it is true reflection of past behavior of peers in network. The calculation of global trust is done recursively and it converges at some unique value. With the suitable choice of parameter $p$ and $q$, (Eq.\ref{equ0}) the algorithm converges much faster. The updates have to be send about only those peers whose local trust value is changing. Hence, lesser number of messages are required to update the global trust. This algorithm can be implemented in truly distributed system where no central authority is present. We have presented the results for simulation  and it shows that this algorithm is robust against the various attacks like individual malicious, unpredictable malicious and collective malicious. Lastly we have shown through simulations that, because peers are not competing with each others for higher value of global trust hence load on individual good peer is more uniform compared to when relative ranking mechanism is used.
\par In future we will look into how to handle the Sybil attack in an  efficient way. The choice of parameter  $p$ and $q$ are very important and a trade-off is involved in it. The appropriate values of $p$ and $q$ need further investigation.

\begin{IEEEbiography}{Sateesh Kumar Awasthi}%[{\includegraphics[width=1in,height=1.25in,clip,keepaspectratio]%{sateesh.png}}]{Sateesh Kumar Awasthi}
He was born in Uttarkashi, India. He is currently pursuing Ph.D in the Department of Electrical Engineering at  IIT, Kanpur. His research interest include Peer-to-Peer Networks, Wireless Sensor Networks, Complex Networks, Social Networks, Solution of non-linear equations, Application of Linear Algebra and Game theory in Networks. 
\end{IEEEbiography}
\begin{IEEEbiography}{Yatindra Nath Singh}
%[{\includegraphics[width=1in,height=1.25in,clip,keepaspectratio]{y-n-singh.png}}]{Yatindra Nath Singh}
He was born in Delhi, India. He was awarded Ph.D for his work on optical amplifier placement problem in all-optical broadcast networks in 1997 by IIT Delhi. In July 1997, he joined EE Department, IIT Kanpur. He was given AICTE young teacher award in 2003. Currently, he is working as professor. He is fellow of IETE, senior member of IEEE and ICEITE, and member ISOC. He has interests in telecommunications' networks specially optical networks, switching systems, mobile communications, distributed software system design. He has supervised 7 Ph.D and more than 97 M.Tech theses so far. He has filed three patents for switch architectures, and have published many journal and conference research publications. He has also written lecture notes on Digital Switching which are distributed as open access content through content repository of IIT Kanpur. He has also been involved in opensource software development. He has started Brihaspati (brihaspati.sourceforge.net) initiative, an opesource learning management system, BrihaspatiSync – a live lecture delivery system over Internet, BGAS – general accounting systems for academic 
institutes.

\end{IEEEbiography}

%\begin{IEEEbiographynophoto}{Yatindra Nath Singh}

%\newpage



%\vfill


\end{document}